\documentclass[conference,letterpaper,onecolumn]{IEEEtran}

\usepackage{float}
\usepackage{caption}
\usepackage{subcaption}
\usepackage[utf8]{inputenc} 
\usepackage[T1]{fontenc}
\usepackage{url}
\usepackage{ifthen}
\usepackage{graphicx}
\usepackage{amsthm}
\usepackage{amsfonts}
\usepackage{cite}
\usepackage{multicol}
\usepackage{color}
\usepackage{amssymb}
\usepackage[cmex10]{amsmath} 
\usepackage{stfloats}
\usepackage{enumitem}
\newlist{steps}{enumerate}{1}
\setlist[steps, 1]{label = Step \arabic*:}
\usepackage{comment}

\newtheorem{theorem}{Theorem}[section]

\newtheorem{defn}{Definition}
\newtheorem{thm}{{\cal T}heorem}
\newtheorem{cor}{Corollary}
\newtheorem{prop}{Proposition}
\newtheorem{lem}{Lemma}
\newtheorem{conj}{Conjecture}
\newtheorem{constr}{Construction}
\newtheorem{note}{Note}
\newtheorem{example}{Example}
\newcommand{\bit}{\begin{itemize}}
	\newcommand{\eit}{\end{itemize}}
\newcommand{\bcor}{\begin{cor}}
	\newcommand{\ecor}{\end{cor}}
\newcommand{\beq}{\begin{equation}}
	\newcommand{\eeq}{\end{equation}}
\newcommand{\beqn}{\begin{equation}}
	\newcommand{\eeqn}{\end{equation}}
\newcommand{\bea}{\begin{eqnarray}}
	\newcommand{\eea}{\end{eqnarray}}
\newcommand{\bean}{\begin{eqnarray*}}
	\newcommand{\eean}{\end{eqnarray*}}
\newcommand{\ben}{\begin{enumerate}}
	\newcommand{\een}{\end{enumerate}}
\newcommand{\bdefn}{\begin{defn}}
	\newcommand{\edefn}{\end{defn}}
\newcommand{\bnote}{\begin{note}}
	\newcommand{\enote}{\end{note}}
\newcommand{\bprop}{\begin{prop}}
	\newcommand{\eprop}{\end{prop}}
\newcommand{\blem}{\begin{lem}}
	\newcommand{\elem}{\end{lem}}
\newcommand{\bthm}{\begin{thm}}
	\newcommand{\ethm}{\end{thm}}
\newcommand{\bconj}{\begin{conj}}
	\newcommand{\econj}{\end{conj}}
\newcommand{\bconstr}{\begin{constr}}
	\newcommand{\econstr}{\end{constr}}
\newcommand{\bpf}{\begin{proof}}
	\newcommand{\epf}{\end{proof}}

\newtheorem{remark}{Remark}
\IEEEoverridecommandlockouts
\begin{document}
	\title{On Existence of Latency Optimal Uncoded Storage Schemes in Geo-Distributed Data Storage Systems}
	\author{
		\IEEEauthorblockN{Srivathsa Acharya, P. Vijay Kumar \ \\}
		\IEEEauthorblockA{
			Department of Electrical Communication Engineering,\\
			IISc., Bangalore \\ 
			Email: \{srivathsaa, pvk \}@iisc.ac.in}
		\and
        \IEEEauthorblockN{Viveck R. Cadambe \ \\}
		\IEEEauthorblockA{
			Department of Electrical Engineering,\\
			Pennsylvania State University, USA \\ 
			Email: viveck@psu.edu}
		\thanks{This research is supported by SERB  Grant No.~CRG/2021/008479 and NSF Grant~\#2211045.}
    }
	\maketitle
	\begin{abstract}
We consider the problem of geographically distributed data storage in a network of servers (or nodes) where the nodes are connected to each other via communication links having certain round-trip times (RTTs). Each node serves a specific set of clients, where a client can request for any of the files available in the distributed system. The parent node provides the requested file if available locally; else it contacts other nodes that have the data needed to retrieve the requested file. This inter-node communication incurs a delay resulting in a certain latency in servicing the data request. The worst-case latency incurred at a servicing node and the system average latency are important performance metrics of a storage system, which depend not only on inter-node RTTs, but also on how the data is stored across the nodes. Data files could be placed in the nodes as they are, i.e., in uncoded fashion, or can be coded and placed. This paper provides the necessary and sufficient conditions for the existence of uncoded storage schemes that are optimal in terms of both per-node worst-case latency and system average latency. 
In addition, the paper provides efficient binary storage codes for a specific case where optimal uncoded schemes do not exist. 
\end{abstract}

\section{Introduction}
\label{sec:intro}
Distributed data storage systems are an integral part of modern cloud-computing infrastructure. Over the last decade, coding theory has played an integral role in ensuring cost-effective fault-tolerance for distributed data storage systems, for e.g., through the development of regenerating codes \cite{dimakis2010network, wu2009reducing}, locally repairable codes \cite{gopalan2012locality, tamo2014family} and codes with availability \cite{rawat2016locality} (see \cite{balaji2018erasure}, \cite{vinayaketal2022FnT} for a survey). In this paper, we study a coding formulation that is relevant for geographically distributed (or \emph{geo-distributed}) cloud storage systems  where the data is replicated primarily to enable low latency data access to clients across a wide geographic area. In fact, most major commercial cloud storage providers including Google Cloud \cite{Spanner}, Amazon AWS \cite{aws-geo}, and Microsoft Azure \cite{cosmos-geo} offer support for geo-distributed data storage.  
 
Geo-distributed cloud storage systems consist of nodes (data-centers/servers) connected to each other through links having certain round-trip delays. Each node serves a specific set of clients, where each client can request for any data available in the system. One of the desired features of geo-distributed storage systems is to provide wait-free or low-latency access to data.  Providing wait-free access requires every file to be replicated at every node, which is inefficient in terms of storage utilization, and also infeasible when the storage requirement is comparable to the total storage capacity of the system. Given that total replication  of files at nodes is not possible, several schemes based on \emph{partial replication} have been proposed in the literature, where each node stores only a subset of the data (see \cite{cadambe},\cite{cadambeArxiv} and references therein), which we refer to  as \emph{uncoded storage schemes}.

In uncoded storage schemes where nodes store only a subset of the data, clients may have to fetch data from remote nodes, and thereby incur a data access latency of the  inter-node round-trip-time (RTT)\footnote{Assuming that the system is well-provisioned,  RTTs between the nodes, which can be relatively large (tens to a few hundreds of ms,  see Sec. \ref{sec:app}) are a dominant component of user data access latency.}. In this paper, we study two latency metrics that are relevant to practice. First, we consider the worst-case latency incurred over the system - the maximum round trip time required to fetch an object from a node for a given storage scheme. The second metric is the average latency (measured across nodes and files), which determines the average throughput of the system  as per Little's law\footnote{The \emph{throughput} of a data store - the average number of client requests that can be served per second - is an important metric in data store design.} \cite{cadambeArxiv}.

Instead of storing copies of data files across nodes, one could also store functions of files (for e.g., linear combination of files) in the nodes, which we refer to as \textit{coded storage}\footnote{Coded storage is known as \emph{erasure coding} in the existing literature.}. Coded storage can be beneficial (in terms of latency) over  uncoded storage schemes in certain cases, while in others, uncoded schemes work best. Both cases are illustrated below.

\begin{example}
Consider a data storage system with 4 nodes $\{A,B,C,D\}$, each capable of storing one file. Suppose the system has to store $3$ information files $\{W_1,W_2,W_3\}$. The nodes along with inter-node RTTs are depicted in Fig.~\ref{fig:eg1}. The figure shows two possible ways of storing the information files on the nodes. In the first method, uncoded files are placed whereas the second method uses coded storage on node $D$ which stores a coded file that is bit-wise XOR of the $3$ files. 
 
\begin{figure}[t]
\centering
\includegraphics[scale=0.4]{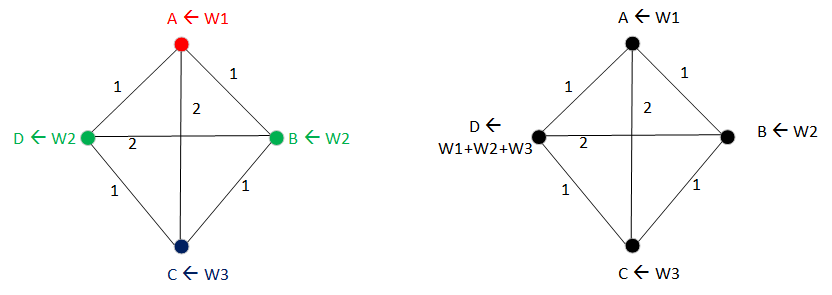}
\caption{Example 1: Data store with $4$ nodes and $3$ files with inter-node RTTs. Storage type - Left: Uncoded,  Right: Coded.}
\label{fig:eg1}
\end{figure}

\begin{table}
\caption{Per-node worst-case latencies and system average latency for the coded and uncoded schemes shown in Fig.~\ref{fig:eg1}. }
\label{tab:eg1}
\centering
\begin{tabular}{||c|c|c|c|c||}
	\hline
	Scheme & Node/ & Decoding &  Worst-case & Average  \\
	 &   Codeword & &  Latency & Latency \\	
	& $i (X_i)$ & & $L_{max}^{(i)}$ & $L_{avg}$\\
	\hline 
	\hline
	&$A (W_1)$& &  2 &\\
	Uncoded&$B (W_2)$ & &  1 & $\frac{10}{12}$  \\
	&$C (W_3)$& &  2 & \\
	&$D (W_2)$& &  1 & \\
	\hline 
	& $A (W_1)$& $W_3 = X_D \oplus X_A \oplus X_B$&   1 & \\
	Coded &$B (W_2)$& &   1 &  $\frac{9}{12}$\\
	&$C (W_3)$& $W_1 = X_D \oplus X_C \oplus X_B$&   1 & \\
	&$D (\oplus_{i=1}^3 W_i)$& $W_2 = X_D \oplus X_A \oplus X_C$&   1 & \\ 
	\hline
\end{tabular}
\vspace{-1em}
\end{table}

\noindent Denote the contents of nodes as $\{X_A, X_B, X_C, X_D\}$. Table \ref{tab:eg1} shows the encoding and decoding of the 3 information files at each node. The resulting per-node worst-case latencies and system-average latency\footnote{The latency terms are formally defined in Section \ref{sec:model}.} are also provided in the table. We see  that the coded scheme outperforms the uncoded scheme in both the latency metrics. It can be further verified that the coded scheme has the least latency over all uncoded schemes.

\end{example}

\begin{example}
Consider again a 4-node 3-file system as in Example 1, but with different inter-node RTTs as given in Fig.~\ref{fig:eg2}. 
\begin{figure}[t]
\centering
\includegraphics[scale=0.4]{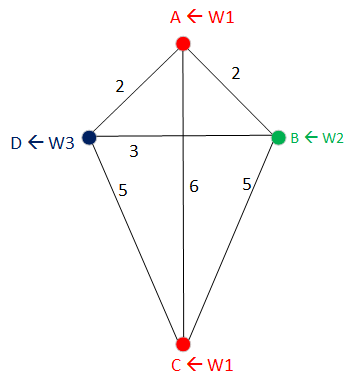}
\caption{Example 2: Data store with $n$ = 4 nodes $\{A,B,C,D\}$, $k$ =3 files $\{W_1, W_2, W_3\}$, and their inter-node RTTs. 
}
\label{fig:eg2}
\vspace{-1em}
\end{figure}
\noindent An uncoded storage scheme is also shown in the figure, with which it is possible for each node to obtain all the 3 files by contacting only its $2$ least RTT neighbors. This results in minimum per-node worst-case latency and average latency which no coded storage scheme can beat. 
\end{example}

Thus, it is useful to know the class of storage systems where an uncoded scheme itself gives optimal worst-case and average latency. This is the focus of current paper.  It is notable that in computer systems and performance analysis literature, there are several works that aim to optimize data placement in geo-distributed data storage systems by utilizing knowledge of inter-node RTTs \cite{ardekani2014self, spanstore, replication-placement, shankaranarayanan2014performance, Abebe2018ECStoreBT, su2016systematic, ec-geo-placement,zare2022legostore,SoljaninEtAl}.  These works develop optimization frameworks and solutions for data/codeword placement based on latency, communication cost, storage budget, and fault-tolerance requirements. However, even for the simpler objective of minimizing average and worst-case latencies, the best strategies are not known. In particular, for a given storage budget and worst-case latency, it is unclear when uncoded strategies obtain optimal average latency, or how erasure codes should be designed to minimize average latency. There are also works which provide latency analsysis based on MDS storage (as in \cite{SoljaninEtAl},\cite{lee_shah_huang_ramachandran_2017}), but as will be shown later, MDS codes are not suited well for average latency constraints. Notably, different from classical erasure codes, the erasure codes must be designed and codeword symbols must be placed on the nodes based on the RTTs to minimize latency. This paper makes progress on these problems.

The rest of the paper is organized as follows. Section \ref{sec:model} develops the system-model, provides formal definitions of storage codes and associated latencies. Section \ref{sec_main_result} gives the main contribution of the paper where the problem of optimal uncoded storage is converted to one of \emph{vertex coloring} on a special subgraph  called the \emph{nearest-neighbor graph}, and provides necessary and sufficient condition for an optimal uncoded scheme to exist. Section \ref{sec:coded storage} provides coded storage schemes for some specific cases where optimal uncoded schemes do not exist.  Section \ref{sec:app} gives an application of the main result on a hypothetical geo-distributed data-center network. The paper concludes with possible research directions for future in Section\ref{sec:conc}.

\section{System Model and Problem formulation}
\label{sec:model}
We model the data storage network of $n$ servers and $k \le n$
files by an undirected weighted complete\footnote{A complete graph is one where an edge exists between every pair of nodes.} 
graph $\mathcal{G}=(\mathcal{N}, T)$, where $\mathcal{N} = \{1,2,\dots, n\}$ denote the $n$ nodes, and $T = \{\tau(i,j): 1\le i,j \le n\}$ is the edge-weight matrix, representing the RTTs between a pair of nodes. RTT is same in either direction, i.e., $\tau(i,j) = \tau(j,i) \quad \forall (i,j) \in \mathcal{N} \times \mathcal{N}$. Also $\tau(i,i) = 0 \quad \forall i \in \mathcal{N}$. 
Let $W_1,\dots W_k$ denote the $k$ information files (each of unit file-size) to be stored in the storage network. Each node has capacity to store data worth $1$ file-size \footnote{In a general setting, each of the $n$ servers can store $M \ge 1$ files, and the requirement is to store $kM (\le nM)$ information files in the network. The paper addresses $M=1$ case. The storage codes thus obtained can be extended to $M >1$ case by partitioning each of the node contents into $M$ \emph{stripes} and then applying the code on each stripe.}.
 We denote $[k] = \{1,2,\dots,k\}$ to represent the file index set. Let $X_1,\dots X_n$ denote the data stored in each of the $n$ nodes. 

Note that the point-to-point single-hop model above can also be applied to a more general multi-hop scenario, where the communication between two nodes $(i,j)$ traverses intermediate nodes. In this case, the sum total of RTTs of the links along the least RTT path between the nodes is taken as the equivalent edge weight $\tau(i,j)$ in our model.

A linear storage code with sub-packetization $\alpha$ can be defined as follows. Assume that each information file $W_j; j \in [k]$ can be split into $\alpha$ sub-packets $(W_{j1}, W_{j2},\dots, W_{j\alpha})$, with each sub-packet belonging to a finite field $\mathcal{F}$. Similarly, a file stored in node $i \in \mathcal{N}$, $X_i$, is composed of $\alpha$ sub-packets  $(X_{i1}, X_{i2},\dots, X_{i\alpha})$, also belonging to $\mathcal{F}$.
Denote 
$\\
\underbar{X} = (X_{11},X_{12}, \dots ,X_{1\alpha}, \dots, X_{n1}, X_{n2},\dots, X_{n\alpha})^T \\ 
\underbar{W} = (W_{11},W_{12}, \dots ,W_{1\alpha}, \dots, W_{k1}, X_{k2},\dots ,W_{k\alpha})^T$.

\begin{defn}[Linear storage code]
A linear storage code $\mathcal{C}$ on $\mathcal{G}$ is defined by a $(k\alpha \times n\alpha)$ generator matrix $G$ such that 
\begin{equation}
\label{eq_coding_G}
\underbar{X}^T = \underbar{W}^T G
\end{equation}
\end{defn}

The column $j$ in  $G$ comprises the linear weights of each of the $k\alpha$ sub-packets of $\underbar{W}$ that combine to make the $j$th coded sub-packet in $\underbar{X}$. We only consider the codes with $rank(G) = \alpha k$, the condition necessary for decoding all information files from the coded files.

An \emph{uncoded storage scheme} is a special case of linear storage codes where there is no sub-packetization ($\alpha=1$) and no coding across the files. 

\subsection{Average Latency and Per-node Worst-case Latency}  
Given a code $\mathcal{C}$ on $\mathcal{G}$, the \emph{decoding process} and associated \emph{latencies} are formally defined in Appendix~\ref{recovery} of the extended paper~\cite{acharyaArxiv}. Here, we provide an intuitive definition of the latency as follows. For a certain wait-time $L$, a node $i$ has access to the contents of those nodes $t$ whose RTT satisfies $\tau(t,i) \le L$. Using the contents from these nodes, certain raw files can be decoded.
\emph{Latency} at node $i$ to decode a file $W_j$, denoted as $l_{j}^{(i)}$, is defined as the minimum wait-time $L$ at node $i$  needed to decode file $W_j$\footnote{Strictly speaking, as users/clients request files from nodes, the latency should also include the delay between a user and its local node. But this delay can be neglected since it is dominated by inter-node RTT (see \cite{spanstore},\cite{zare2022legostore}), and since the delay remains same for any choice of the storage code, thus not affecting the storage optimization.}.

\emph{Per-node worst-case latency} of code $\mathcal{C}$ at node $i$ is defined as 
\begin{equation}
\label{def_wc_lat}
L_{max}^{(i)}(\mathcal{C}) = \max_{j \in [k]} l_j^{(i)}
\end{equation}

\emph{Average latency} of code $\mathcal{C}$ is defined as
\begin{equation}
\label{def_avg_lat}
L_{avg}(\mathcal{C}) = \frac{1}{kn}\sum_{i \in \mathcal{N}}\sum_{j \in [k]}l_{j}^{(i)}
\end{equation}
Given a node $i$, let $ \big( \lambda_0^{(i)} \le \lambda_1^{(i)} \le \dots \le \lambda_{(n-1)}^{(i)} \big)$ be the sorted list of RTTs to node $i$, i.e., $\big (\tau(j,i): j \in \mathcal{N}\big )$, in ascending order. That is, $\lambda_{m}^{(i)}$ is the $m^{\text{th}}$ least RTT value to node $i$ from other nodes.  By definition, $\lambda_0^{(i)}= \tau(i,i) =0$.

\begin{prop}
\label{prop:wc_bound}
For any code $\mathcal{C}$ on $\mathcal{G}$, per-node worst-case latency at any node $i$ is lower-bounded as:
\begin{equation}
\label{eq_wc_bound}
L_{max}^{(i)}(\mathcal{C}) \ge \lambda_{(k-1)}^{(i)}
\end{equation}

Further, the average latency $L_{avg}(\mathcal{C})$ is lower-bounded as:  
\begin{equation}
	\label{eq_avg_bound}
	L_{avg}(\mathcal{C}) \ge \frac{1}{kn}\sum_{i \in \mathcal{N}}\sum_{j\in [k]}\lambda_{j}^{(i)}
\end{equation}

\end{prop}
\begin{proof}
See Appendix~\ref{app:wc_bound} of the extended paper~\cite{acharyaArxiv}.
\end{proof}

\begin{defn}
Given a directed subgraph $\mathcal{D}$ of $\mathcal{G}$ with the same node set, a code $\mathcal{C}$ is said to be \emph{admissible} on $\mathcal{D}$ if the decoding of any file at any node involves file transfers only along the directed edges of $\mathcal{D}$ .  
\end{defn}

\section{Main Result}
\label{sec_main_result}
In this paper, we consider only the codes meeting the worst-case latency optimality constraint \eqref{eq_wc_bound}. Satisfying this constraint are the codes admissible on a special subgraph called the \emph{nearest-neighbor graph}.
\bdefn
A \emph{nearest-neighbor graph} $\mathcal{G}_{k-1}$ is defined as a directed subgraph of $\mathcal{G}$ where each node $i$ has incoming edges from $(k-1)$ other nodes having $(k-1)$ least RTT values to node $i$. 
\edefn

\begin{remark} \label{rem:nng}
 \
\begin{itemize}[leftmargin =*]
\item The incoming edges to a node $i$ in $\mathcal{G}_{k-1}$ have, in ascending order, the weights $\lambda_1^{(i)}, \dots , \lambda_{(k-1)}^{(i)}$.
\item In this paper, we refer to \textbf{neighbors} of a node $i$ in $\mathcal{G}_{k-1}$  as only the $(k-1)$ nodes from which there are incoming edges to node $i$. But, a node connected only via an outgoing edge from node $i$ is \textbf{not} referred as its \emph{neighbor} in $\mathcal{G}_{k-1}$.
\item Multiple $\mathcal{G}_{k-1}$ are possible when multiple nodes share the same RTT value of $\lambda_{(k-1)}^{(i)}$ to a node $i$. 
\end{itemize}
    
\end{remark}

\noindent $\mathcal{G}_{k-1}$ for Examples 1 and 2 of Section \ref{sec:intro} are shown in Fig.~\ref{fig:egBng}.

\begin{figure}[t]
\centering
\begin{subfigure}{0.45\columnwidth}
\centering
\includegraphics[scale = 0.4]{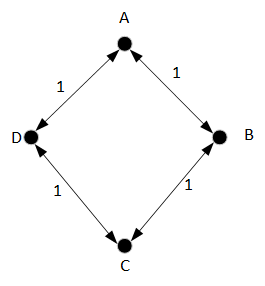}
\caption{Example 1}
\end{subfigure}
\begin{subfigure}{0.45\columnwidth}
\centering
\includegraphics[scale = 0.4]{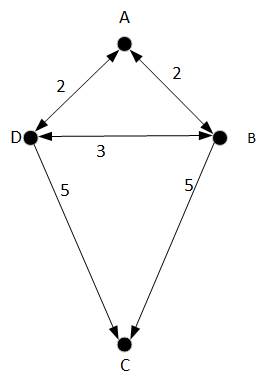}
\caption{Example 2}
\end{subfigure}
\caption{Nearest-neighbor graphs $\mathcal{G}_{2}$ for Section \ref{sec:intro} examples.}

\label{fig:egBng}
\vspace{-1em}
\end{figure}

\noindent Note that there always exist codes that are admissible on $\mathcal{G}_{k-1}$, such as the MDS (Maximum Distance Separable) codes given below.
\begin{example}  [Scalar MDS Codes]
\noindent Let $\alpha = 1$(no sub-packetization) \footnote{\emph{Storage code without sub-packetization} is known as \emph{Cross-object erasure coding} in \cite{cadambe}. MDS codes with sub-packetization also exist (see \cite{cadambeArxiv}).} and $G$ be a $k\times n$ matrix of a MDS code. From $k$ information files, let $n$ coded files be generated with this MDS matrix $G$ as in \eqref{eq_coding_G}, and place one coded file at each node. Due to MDS property, any information file can be recovered at a given node if there are $k$ coded files, which can be obtained by using the local data at the node and by contacting $(k-1)$ least RTT nodes. Hence this code is admissible on a $\mathcal{G}_{k-1}$.
    
\end{example}

\noindent The reason for looking into admissible codes on $\mathcal{G}_{k-1}$ is because of their latency optimality properties as shown next.

\bprop
\label{prop:wc_bng}
Any admissible storage scheme $\mathcal{C}$ on  $\mathcal{G}_{k-1}$ meets the per-node worst-case latency bound in \eqref{eq_wc_bound}, i.e,
\begin{equation}
L_{max}^{(i)}(\mathcal{C}) = \lambda_{(k-1)}^{(i)}
\end{equation}
  
\eprop
\begin{proof}
A given node $i$ can only use the links available in $\mathcal{G}_{k-1}$, and the maximum weight of its incoming edges is $\lambda_{(k-1)}^{(i)}$. 
\end{proof}

Even though any admissible code (such as MDS code) on $\mathcal{G}_{k-1}$ is worst-case latency optimal, it need not be average-latency optimal.
However, if an \textbf{uncoded} admissible code exists on a $\mathcal{G}_{k-1}$, it is both worst-case latency and average-latency optimal as given below.
 
\bprop
\label{prop: uncoded_bng}
If there is an admissible \textbf{uncoded} scheme $\mathcal{C}$ on $\mathcal{G}_{k-1}$, then it meets the average latency lower bound in  \eqref{eq_avg_bound}: 
\begin{equation} \label{eq_opt_uncoded}
L_{avg}(\mathcal{C}) = \frac{1}{kn}\sum_{i \in \mathcal{N}}\sum_{j\in [k]}\lambda_{j}^{(i)}
\end{equation}
Conversely, an optimal uncoded scheme in the original complete graph $\mathcal{G}$ that meets the latency bounds of \eqref{eq_wc_bound} and  \eqref{eq_avg_bound} exists only if it is admissible on \emph{some} $\mathcal{G}_{k-1}$ of $\mathcal{G}$.
\eprop
\begin{proof}
See Appendix~\ref{app:uncoded_bng} of the extended paper~\cite{acharyaArxiv} for proof of \eqref{eq_opt_uncoded}. The converse can be seen from the fact that every node is forced to communicate with precisely $(k-1)$ least latency neighbors in order to meet both the latency bounds.
\end{proof}

We thus look for existence of admissible uncoded storage schemes on $\mathcal{G}_{k-1}$. For this, we first convert the problem into that of vertex coloring\cite{diestel} on a related undirected graph called the \emph{extended graph}. 

\bdefn
Given a nearest-neighbor graph $\mathcal{G}_{k-1}$, its extended graph $\mathcal{H}$ is defined as an undirected graph on same node set formed by the following rules. 
\begin{itemize}[leftmargin=*]
\item If 2 nodes are connected by a  (directed) edge in $\mathcal{G}_{k-1}$, connect them by an (undirected) edge in $\mathcal{H}$.
\item If 2 nodes are neighbors\footnote{see Remark~\ref{rem:nng} for definition of neighbor in $\mathcal{G}_{k-1}$} of same node in $\mathcal{G}_{k-1}$, then also connect them by an edge in $\mathcal{H}$.
\end{itemize}
\edefn
\noindent Fig.~\ref{fig:egH} shows the extended graphs for the examples of  Section~\ref{sec:intro}.
\begin{figure}[t]
\centering
\begin{subfigure}[b]{0.45\columnwidth}
\centering
\includegraphics[scale = 0.5]{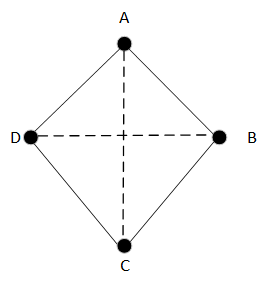}
\caption{Example 1}
\end{subfigure}
\begin{subfigure}[b]{0.45\columnwidth}
\centering
\includegraphics[scale = 0.5]{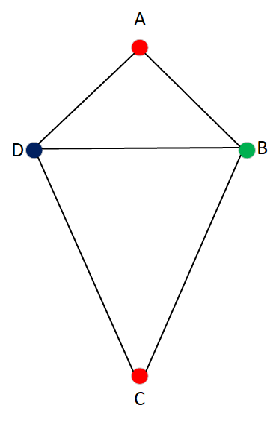}
\caption{Example 2}
\end{subfigure}
\caption{Extended graphs for the examples in Section \ref{sec:intro}.The dashed edges are those added on top of $\mathcal{G}_{2}$.}
\label{fig:egH}
\vspace{-1em}
\end{figure}

Vertex coloring of a graph $\mathcal{G} = (\mathcal{N},\mathcal{E})$ is a map $\rho: \mathcal{N} \to S$ such that $\rho(v) \ne \rho(w)$ whenever $v$ and $w$ are adjacent. The elements of set $S$ are called the \emph{colors}. 
The smallest size of set $S$ with which $G$ can be vertex colored is known as its \emph{chromatic number} denoted by $\chi(\mathcal{G})$. 
One result that we use is that if $\mathcal{G}$ has a complete subgraph of $m$ nodes, then $\chi(\mathcal{G}) \ge m$.

\begin{theorem}[Vertex Coloring]
\label{thm:vc}
An admissible uncoded storage scheme exists on $\mathcal{G}_{k-1}$ if and only if the corresponding extended graph $\mathcal{H}$ has chromatic number $\chi(\mathcal{H})= k$ . 
\end{theorem}
\begin{proof} 
    By associating each file with distinct color, result is obtained. A detailed proof is in Appendix~\ref{app:vc} of the extended paper~\cite{acharyaArxiv}.
\end{proof}

\noindent Thus, the theorem along with Proposition~\ref{prop: uncoded_bng} implies that a latency optimal uncoded scheme exists on the original graph $\mathcal{G}$ if and only if $\mathcal{H}$ of some $\mathcal{G}_{k-1}$ is $k$-colorable. Let us apply this result to $(n,k) = (4,3)$ systems of Section~\ref{sec:intro}.
\begin{itemize}[leftmargin =*]
\item
\noindent \textbf{Example 1:} Extended graph is itself a complete graph of $4$ nodes, and hence it needs $4 (>k)$ colors. So, from Theorem \ref{thm:vc}, no latency optimal uncoded scheme exists, which reinforces the observation in Section \ref{sec:intro}.
\item
\noindent \textbf{Example 2:} By assigning same color to the non-adjacent nodes $A$ and $C$ (in $\mathcal{H}$), $k=3$ coloring is possible as shown in Fig.~\ref{fig:egH} . Thus, an optimal uncoded scheme exists as shown in Section \ref{sec:intro}. \\
\end{itemize}
\vspace{-1em}
\noindent One consequence of the theorem is for special case of $k=2$.
\bcor
\label{corr:k_2}
For any data storage system $\mathcal{G}$ with $k=2$, there always exists an optimal uncoded scheme that meets the latency bounds of \eqref{eq_wc_bound} and \eqref{eq_avg_bound}.
\ecor
\begin{proof}
 See Appendix~\ref{app:k_2} of the extended paper~\cite{acharyaArxiv}.    
\end{proof}
%
%
%

\section{Coded Storage Schemes on $\mathcal{G}_{k-1}$}
\label{sec:coded storage}

\noindent We next look at networks where Theorem \ref{thm:vc} does not hold on any nearest-neighbor graph $\mathcal{G}_{k-1}$. As no optimal uncoded scheme exists, we need to look for coded schemes on $\mathcal{G}_{k-1}$ that are average latency optimal. This is an open problem. However, as a byproduct of Theorem \ref{thm:vc}, we provide a family of admissible binary codes on $\mathcal{G}_{k-1}$ for the special case of $\chi(\mathcal{H}) = (k+1)$ as described below.


Consider a data storage network where an extended graph has $\chi(\mathcal{H}) = (k+1)$, i.e., $\mathcal{H}$ needs one more color than the no.~of files. By replacing $k$ of these colors by files, and then  converting an extra color to an appropriate linear combination of files, it is possible to get an admissible coding scheme on $\mathcal{G}_{k-1}$, as described next.

$\chi(\mathcal{H}) = (k+1)$ implies a valid $(k+1)$ vertex-coloring map: $\rho: \mathcal{N} \to S:= \{c_1, c_2, \dots, c_{k+1}\}$. Associate some $k$ of the $(k+1)$ colors directly with $k$ file indices , and mark the remaining color as \emph{coded}. That is, form a bijective function $f:S \to \{1,2,\dots,k,*\}$ where $*$ represents a coded color. 

\subsection{Encoding Algorithm at each node $i \in \mathcal{N}$} \label{sec:binary code}

If $i$ is mapped to an \emph{uncoded} color, i.e, $f(\rho(i)) \in [k]$, then the file assignment is $X_i = W_{f(\rho(i))}$.
On the other hand, if $i$ is mapped to a coded color $f(\rho(i))=*$, then:
\begin{itemize}[itemindent =*, leftmargin =*]
    \item Let $\mathcal{R}(i) = \{j\in \mathcal{N}: (i,j) \in \mathcal{G}_{k-1}\}$ be the set of nodes that are adjacent to node $i$ via outgoing edges from $i$ in $\mathcal{G}_{k-1}$ (called as \emph{receive} nodes).
    \item For each receive node $r \in \mathcal{R}(i) $, identify the index of the missing file $\mu(r)$ as follows. The node $r$ and all of its $(k-1)$ neighbors in $\mathcal{G}_{k-1}$ except node $i$ have uncoded colors. This is because these nodes are adjacent to $i$ in $\mathcal{H}$ and hence cannot share the same color as $i$.  Therefore, $r$ has access to some $(k-1)$ uncoded files from its $(k-2)$ uncoded neighbors and itself. Hence, there is precisely one file which is not available with $r$ that it wishes to get from $i$. Denote the index of this missing file as $\mu(r)$ 
    \item For node $i$, its $(k-1)$ neighbors in $\mathcal{G}_{k-1}$  have distinct, uncoded colors due to valid vertex coloring. Hence $i$ has a missing file which needs to be provided by itself. Denote this missing file index as $\mu(i)$  
    \item Assign the sum (bitwise-XOR) of missing files to node $i$ as:
\begin{equation}\label{eq_index_coding}
    X_i = \sum_{f \in S} W_{f} \text{  } ,  S:= \big \{\mu(r): r \in \{i\}\cup \mathcal{R}(i)\}\big \}      
\end{equation}
\end{itemize} 
By construction, the above code is admissible as every node has at most one missing uncoded file, which can be obtained from its coded neighbor. For clarity, a \emph{decoding} algorithm has been added in Appendix~\ref{app:binary code} of the extended paper~\cite{acharyaArxiv}. 


Using the above algorithm, we can get multiple admissible codes, one for each choice of vertex coloring on $\mathcal{H}$ (unique up to color permutation), and for each choice of coded color.  It is not known whether this family of codes contains a system average-latency optimal code on $\mathcal{G}_{k-1}$. Nevertheless, the codes are attractive from implementation perspective since they are worst-case latency optimal, binary-coded, and have just enough file additions to make the code admissible. 

\section{Application to a Data-Storage System}
\label{sec:app}
\noindent We illustrate the efficacy of Theorem~\ref{thm:vc}  on a sample geo-distributed data-center network of $6$ nodes as shown in Fig.~\ref{fig:awsdc}. The inter-node RTTs are taken from measurements as per Amazon AWS public cloud\cite{cadambe}\cite{aws}. Consider $k=4$ files.
\begin{figure}[t]
\centering
\begin{subfigure}[b]{\columnwidth}
\centering
\includegraphics[scale = 0.8]{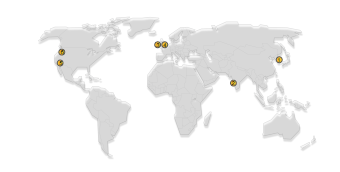}
\end{subfigure}
\\
\begin{subfigure}[b]{\columnwidth}
\centering
\begin{tabular}{|l|cccccc|}
\hline
Regions & Seoul & Mumbai & Ireland & London & California & Oregon\\
\hline 
Seoul & 0 & 120 & 230 & 240 &138 & 126\\
Mumbai & 120 & 0 & 121 & 113 &228 & 220\\
Ireland & 230 & 121 & 0 & 13 &138 & 126\\
London & 240 & 113 & 13 & 0 &146 & 137\\
California & 138 & 138 & 230 & 146 &0 & 22\\
Oregon & 126 & 220 & 126 & 137 &22 & 0\\
\hline
\end{tabular}
\end{subfigure}

\caption{A sample data store with $6$ nodes and their inter-node RTTs (in ms) measured as per AWS public cloud\cite{cadambe}\cite{aws}.}
\label{fig:awsdc}
\vspace{-1em}
\end{figure}

\noindent \textbf{Non-existence of Optimal Uncoded Scheme:} 
For $(n,k) = (6,4)$, there is a unique nearest-neighbor graph $\mathcal{G}_3$ as shown in Fig.~\ref{fig:awsk4}. 
\begin{figure}[t]
\centering
\begin{subfigure}[b]{0.45\columnwidth}
\centering
\includegraphics[scale = 0.3]{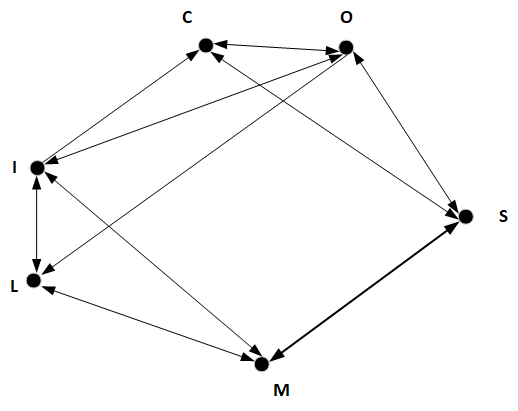}
\caption{(Unique) $\mathcal{G}_3$ }
\end{subfigure}
\begin{subfigure}[b]{0.45\columnwidth}
\centering
\includegraphics[scale = 0.3]{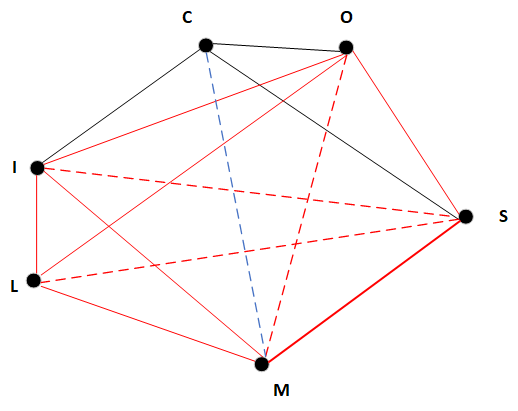}
\caption{The extended Graph $\mathcal{H}$. 
}
\end{subfigure}
\caption{AWS data-store (cities replaced by initials) for $k=4$. 
	}
\label{fig:awsk4}
\end{figure}
The figure also shows the extended Graph $\mathcal{H}$. It can be seen that the induced sub-graph among nodes $S,L,I,M,O$ is a complete graph $\mathcal{K}_5$, and hence $k=4$ coloring is not possible. Thus, as per Theorem \ref{thm:vc} and Proposition \ref{prop: uncoded_bng}, no optimal uncoded scheme exists (in terms of both average and per-node worst-case latency). 

\noindent \textbf{Admissible Binary Codes:}
From Fig.~\ref{fig:awsk4}, we can see that $\mathcal{H}$ can be $5$- colored by assigning same color to $C$ and $L$. Since $\chi(\mathcal{H}) = (k+1)$, we can obtain admissible binary codes as described in Section \ref{sec:coded storage}.
\begin{table}
\caption{AWS data store with $(n,k) = (6,4)$: A Binary-code obtained using $5$-vertex coloring of $\mathcal{H}$. $^*$: coded color.}
\label{tab:aws_k_4}
\centering
\begin{tabular}{||c|c|c|c|cccc||}
\hline
\hline
Node & Vertex &Code & Decoding &\multicolumn{4}{|c||}{File recovery}  \\
Index &  Color& &  &\multicolumn{4}{|c||}{latency (ms)}  \\
$i$ & $\rho(i)$&$X_i$ & &$W_1$ & $W_2$ & $W_3$ & $W_4$ \\
\hline 

(S)& $^*c_1$& $W_1+W_2+W_4$& $W_2 = X_S+X_M + X_O$& 120 & 126 & 138 & 126\\
(M)& $c_2$ &$W_1$ & $W_4 = X_S+X_M+X_I$&  0& 121 & 113 &121\\
(I) & $c_3$& $W_2$& &  121 & 0 & 13 & 126 \\
(L) & $c_4$ &$W_3$& &  113 & 13 & 0 & 137 \\
(C)& $c_4$  &$W_3$&$W_1 = X_S+X_I+X_O$ & 138 & 138 & 0 & 22\\
(O) & $c_5$&$W_4$& $W_1 = X_S+X_I+X_O$&  126 & 126 & 22 & 0 \\
\hline
\hline
\end{tabular}
\vspace{-1em}
\end{table}

Let us mark $c_1$ as the coded color and associate rest of the colors directly with uncoded files as given in Table  \ref{tab:aws_k_4}. Now, $c_1$ is mapped to node $S$(Seoul) which is a neigbhor of  nodes $M$, $C$, and $O$ in $\mathcal{G}_3$ . The node $S$  has a missing file $W_2$ while node $M$ needs $W_4$ and nodes $C,O$ both need $W_1$ from node $S$. Hence the codeword on node $S$ as per \eqref{eq_index_coding} is $X_1 = W_1 \oplus W_2 \oplus W_4$. The decoding equations at each node and the corresponding latencies are also tabulated in Table \ref{tab:aws_k_4}. 
The resulting average latency is $L_{avg}(\mathcal{C})  =81.67$ \textit{ms}, compared to the non-achievable lower bound of $76.37$ \textit{ms}. 

\section{Conclusion and Future Work}
\label{sec:conc}
\noindent We introduced the problem of latency optimal storage schemes on a geo-distributed data network having certain inter-node round-trip times. By modeling the storage network as a weighted complete graph, we showed that a latency optimal uncoded storage exists if and only if it is admissible on a subgraph called the nearest-neighbor graph. We then obtained vertex-coloring based condition for such an optimal uncoded scheme to exist.
In the networks where the vertex coloring condition fails, our result provides justification for employing coded storage. Finding optimal codes for such networks is an open problem and is the direction of our future work. 


\section*{Acknowledgment}
\noindent We thank Dr.~Sridhar Ramesh, {MaxLinear Inc.} for helpful comments on this work. 

\clearpage
\bibliographystyle{IEEEtran}
\bibliography{LatencyOptimalUncodedStorageArXiv}

\clearpage
\appendix
\subsection{ Statement and Proof of Proposition \ref{prop:wc_bound}} \label{app:wc_bound}
For any admissible code $\mathcal{C}$ on $\mathcal{G}$, per-node worst-case latency at any node $i$ is lower-bounded as:
\[ L_{max}^{(i)}(\mathcal{C}) \ge \lambda_{(k-1)}^{(i)}
\]
Further, the average latency $L_{avg}(\mathcal{C})$ is lower-bounded as:  
\[
 L_{avg}(\mathcal{C}) \ge \frac{1}{kn}\sum_{i \in \mathcal{N}}\sum_{j\in [k]}\lambda_{j}^{(i)}
 \]
\begin{proof}

\begin{enumerate}
\item Worst-case latency bound:
For a given code $\mathcal{C}$, for a node $i$, let $r$ be the number of nodes (including $i$) with RTT  $ \le L_{max}^{(i)}(\mathcal{C})$. Now, since $\lambda_{(r-1)}^{(i)}$ is the $(r-1)^{\text{th}}$ least RTT value to node $i$, we must have $L_{max}^{(i)}(\mathcal{C})\ge \lambda_{(r-1)}^{(i)}$ . 
\\ 
Let $X_{i1}, X_{i2},\dots,X_{ir}$ be the data stored in these nodes. Since these are the only files are available for decoding, we should be able to retrieve all the $k$ information files $\{W_1,\dots W_k\}$ from these $r$ files. Thus we need $r\ge k$, which implies $ \lambda_{(r-1)}^{(i)} \ge \lambda_{(k-1)}^{(i)}$. 
\\
Combining these two inequalities, we get 
$L_{max}^{(i)}(\mathcal{C})\ge \lambda_{(k-1)}^{(i)}$. 
\item Average latency bound:
The average latency from \eqref{def_avg_lat} is defined as  
$L_{avg}(\mathcal{C}) = \frac{1}{kn}\sum_{i \in \mathcal{N}}\sum_{j=1}^{k}l_{j}^{(i)} $. \\
Consider any node $i$. For ease of notation, let the sorted sequence of latencies at node $i$ be also $\{l_1^{(i)},l_2^{(i)},\dots, l_k^{(i)}\}$. That is, for a given $m \in [k]$, $l_m^{(i)}$ is the latency incurred to obtain $m^{\text{th}}$ file.  Let $r$ be the number of nodes with RTT $\le l_m^{(i)}$ from $i$. Thus $l_m^{(i)} \ge \lambda_{(r-1)}^{(i)}$.  Since these $r$ nodes can help decode $m$ files, it is necessary that $r\ge m$.
Therefore,
$l_m^{(i)} \ge \lambda_{(r-1)}^{(i)} \ge \lambda_{(m-1)}^{(i)}$.
\\Averaging this inequality over $m \in [k]$ and nodes $i$, we get \\
\[
L_{avg}(\mathcal{C}):= \frac{1}{kn}\sum_{i \in \mathcal{N}}\sum_{m \in [k]}l_{m}^{(i)} \ge \frac{1}{kn}\sum_{i \in \mathcal{N}}\sum_{m \in [k]}\lambda_{(m-1)}^{(i)}
\] 
The assumption of sorted latencies can be removed since latencies across all nodes and all files are averaged in the \emph{LHS} of the equation above.
\end{enumerate}
\end{proof}
	
\subsection{Statement and Proof of Proposition \ref{prop: uncoded_bng}} \label{app:uncoded_bng}
If there is an admissible uncoded scheme $\mathcal{C}$ on $\mathcal{G}_{k-1}$, then it meets the average latency lower bound in  \eqref{eq_avg_bound}: 
\[
L_{avg}(\mathcal{C}) = \frac{1}{kn}\sum_{i \in \mathcal{N}}\sum_{j\in [k]}\lambda_{j}^{(i)}
\]
\begin{proof}
On $\mathcal{G}_{k-1}$, a node $i$ has exactly $(k-1)$ neighbors. Since the code is admissible, $k$ files should be decodable from the local file, and $k-1$ files obtained from the neighbors. As the scheme is uncoded, the $k$ information files have to be stored uncoded in $k-1$ neighbors and node $i$. This means that the uncoded files in node $i$ and its $k-1$ neighbors are distinct. Thus the latencies in fetching the $k$ files are precisely $\lambda_0^{(i)}, \dots , \lambda_{(k-1)}^{(i)}$. Substituting the latencies in \eqref{def_avg_lat}, we meet the lower bound.    
\end{proof}

\subsection{Statement and Proof of Theorem \ref{thm:vc}} \label{app:vc}
An admissible uncoded storage scheme exists on $\mathcal{G}_{k-1}$ if and only if the corresponding extended graph $\mathcal{H}$ has chromatic number $\chi(\mathcal{H})= k$, . 
\begin{proof}
Suppose $\chi(\mathcal{H}) = k$. Then a $k$-coloring scheme on $\mathcal{H}$ implies that in $\mathcal{G}_{k-1}$, each node and its $(k-1)$ neighbors have all distinct colors. Replacing $k$ colors with $k$ files gives an admissible uncoded scheme as each node has access to $k$ distinct files from itself and its  $(k-1)$ neighbors.

Conversely, if an admissible uncoded storage scheme exists on  $\mathcal{G}_{k-1}$, each node $i$ needs to get access to $k$ files either locally, or from its neighbors. Since there are only $(k-1)$ neighbors for each node in $\mathcal{G}_{k-1}$, the files in the node $i$ and its neighbor set should all be distinct. Thus, in extended graph $\mathcal{H}$, any 2 neighbors store different files since, in $\mathcal{G}_{k-1}$, they are either neighbors or belong to the neighbor set of a same node. The files therefore form a $k$-coloring on $\mathcal{H}$.
\end{proof}

\subsection{Statement and Proof of Corollary \ref{corr:k_2}} \label{app:k_2}
For any data storage system $\mathcal{G}$ with $k=2$, there always exists an uncoded scheme which is optimal both in per-node worst-case latency and average latency.
\begin{proof}
We first claim that there exists a nearest-neighbor graph $\mathcal{G}_1$ without any loops.\\
	Note that $\mathcal{G}_1$ is obtained by connecting each node to its least latency neighbor. Suppose there is a directed loop  of nodes $(v_1 \to v_2 \to, ...., v_r\to v_1)$, then the corresponding edge weights satisfy \\
	$\tau_{v_1, v_2} \le \tau_{v_2,v_3} \le \dots \le \tau_{v_{(r-1)}, v_r}\le \tau_{v_r, v_1}$\\
	This implies that $\tau_{v_1, v_2} \le \tau_{v_r, v_1}$.
	Since $v_r$ is the least-latency neighbor of $v_1$, it must be that  $\tau_{v_2, v_1} = \tau_{v_r, v_1}$. Thus all the edges in the loop have same weight. So we can break the loop by assigning $v_2$ as the nearest neighbor of $v_1$ instead of $v_r$. In this manner, all loops can be removed to obtain a loop-less $\mathcal{G}_1$.
	
	Now, for $k=2$, note that $\mathcal{G}_1$ and its extended graph $\mathcal{H}$ are the same, and so $\mathcal{H}$ is loop-free. Since a loop-free graph is a tree or a forest which can always be $2$-colored, $\mathcal{H}$ has $(k=2)$ coloring.  Hence, by Theorem \ref{thm:vc}, optimal uncoded scheme exists.
    
\end{proof}

\subsection{Decoding Algorithm for Coding Scheme in Section \ref{sec:binary code}} \label{app:binary code}
\bit[leftmargin =*]
	\item At each node $r \in \mathcal{N}$, if the requested file is available directly from one of its $(k-1)$ neighbors, then we are done. 
	\item Else, the requested file is, by definition, the missing file $\mu(r)$ (see the encoding algorithm for details). 
	\bit[leftmargin =*]
		\item In this case, the node $r$ or exactly one of its neighbors will have a coded file. Denote this \emph{transmit} node by $t$. Its coded file is given by \eqref{eq_index_coding} as:
		\begin{equation}
		X_t = \sum_{f \in S} W_{f} \text{ where }  S= \big \{\mu(j): j \in \{t\}\cup \mathcal{R}(t)\}\big \} 
        \end{equation}
		
		\item  Node $r$ will first collect the files $\{W_{\mu(j)}: j \in \{{t} \cup \mathcal{R}(t)\}, j\ne r\}$ from its uncoded neighbors. It will then subtract/XOR these files out of $X_t$ to obtain the desired file $W_{\mu(r)}$ as
		\begin{equation}
			W_{\mu(t)} = X_{t} -  \sum_{f \in S, f\ne r}  W_{f} 
		\end{equation}
	\eit	
\eit

\subsection{Decoding Process of a Linear Storage Code} \label{recovery}
 Let $G$ be the generator matrix of a linear code $\mathcal{C}$ on the data storage network $\mathcal{G}$ of $n$ serves and $k$ files as described in Section \ref{sec:model}. Therefore, the coded file vector is given by \[\underbar{X}^T = \underbar{W}^T G \]For simplicity, consider no-subpacketization $\alpha = 1$ case.
A node $i \in \mathcal{N}$ uses certain linear combinations of the stored files to retrieve each information file. In other words, it applies a $n \times k $ matrix $R^{(i)}$ such that 
\[ 
\underbar{X}^TR^{(i)} = \underbar{W}^T \iff
GR^{(i)} = I_{k\times k}. 
\]
We shall call $R^{(i)}$ as the \emph{recovery matrix} for node $i$.

Note that $G$ should have rank $k$ in order to be able to recover all the information files. Therefore, there exists an $n\times k$ matrix $\Gamma$ such that $G\Gamma = I_{k\times k}$. Let $H$ be a $(n-k)\times n$ parity check matrix of the code $\mathcal{C}$.
Therefore $GH^T = I_{k \times (n-k)}$. 
\begin{prop}
Recovery matrix of any node $i$ is of the form
\[
R^{(i)} = \Gamma + H^TA^{(i)} \text{ for some } (n-k) \times k \text{ matrix } A^{(i)} 
\]

\end{prop}
\begin{proof}
We have $G(R^{(i)}- \Gamma) = I-I= 0$ which implies that the $k$ columns of $(R^{(i)}- \Gamma)$ belong to the dual code of $\mathcal{C}$. Since the dual code has generator matrix $H$, we have $(R^{(i)}- \Gamma) = H^TA^{(i)}$ for some $(n-k) \times k$ matrix $A^{(i)}$. 
    
\end{proof} 

Note that the row-indices of non-zero elements of a column $j$ in $R^{(i)}$ are precisely the indices of nodes whose contents are used by node $i$ to recover file $W_j$ . The latency in recovering a file $W_j$ at node $i$, denoted by $l_j^{(i)}$,is therefore (defined as) the maximum of the RTTs from these nodes to node $i$.

As seen from the equation above, column $j$ of $A^{(i)}$ determines which of the elements in the corresponding column of the recovery matrix $R^{(i)}$ are non -zero, and hence the latency $l_j^{(i)}$. The optimal decoding process for a given code $\mathcal{C}$ employs optimal choice of each column $j$ of $A^{(i)}$ for every node $i$, so as to minimize $\{l_j^{(i)}, i \in \mathcal{N}, j \in [k]\}$.

\end{document}